\documentclass[11pt]{article}
\usepackage{amssymb,amsfonts,color,amsthm,amsmath}
\usepackage[english]{babel}
\usepackage {hyperref}
\usepackage{graphicx}

\newcommand{\N}{\mathbb{N}}
\newcommand{\Z}{\mathbb{Z}}
\newcommand{\R}{\mathbb{R}}

\newcommand{\supp}{\text{supp}}
\newcommand{\norm}[1]{\left|\left|#1\right|\right|}

\newtheorem{theorem}{Theorem}
\newtheorem{lemma}{Lemma}
\newtheorem{remark}{Remark}
\newtheorem{proposition}[lemma]{Proposition}

\author{Matthew West}
\title{A Pair of Non-Isometric Potentials With the Same Semiclassical Invariants}

\begin{document}
\maketitle
\begin{abstract}
We show that there exist pairs of non-isometric potentials for the 1D semiclassical Schr\"odinger operator whose spectra agree up to $O(h^\infty)$, yet their corresponding eigenvalues differ no less than exponentially.  This result was conjectured by Guillemen and Hezari in \cite{guillemin2012hezari}, where they prove a very similar result, yet cannot remove the possibility of a subsequence $h_k\to 0$ where the ground state eigenvalues may agree.  
\end{abstract}

\section{Introduction}
Consider the semiclassical Schr\"odinger operator $H=-h^2\frac{d^2}{dx^2}+V(x)$, with $h>0$, $V\in C^\infty(\R)$, and $\lim_{|x|\to\infty}V(x)=\infty$.   Schr\"odinger operators of this form admit discrete spectra, with corresponding rank one eigenspaces.  We denote the eigenvalues $E_0<E_1<E_2<\cdots$ and corresponding $L^2$ normalized eigenfunctions $\psi_0,\psi_1,\psi_2,\cdots$.  Note that both the eigenvalues and eigenfunctions depend on the semiclassical parameter $h$,  we will omit this dependence in our notation.  We say the spectra of semiclassical operators $H$ and $H'$ agree up to order $h^\infty$, $\text{spec}(H)=\text{spec}(H')+O(h^\infty)$, if for each $E_k\in \text{spec}(H)$ and $E_k'\in\text{spec}(H')$, and for each $N\in\N$ there exist constants $C_{k,N}$ such that $|E_k-E_k'|\leq C_{k,N}h^N$ for $h>0$.  With this, we are ready to state our main result (note that the potentials of Theorem \ref{theorem:mainResult} are plotted in Figure \ref{fig:Vpm}).

\begin{theorem}
\label{theorem:mainResult}
There exist pairs of non-isometric potentials $V^\pm(x)\in C^\infty(\R)$ with $V(x)\geq 0$ defining Schr\"odinger operators $H^\pm=-h^2\frac{d^2}{dx^2}+V^\pm(x)$ whose semiclassical spectra agree modulo $O(h^\infty)$, yet their eigenvalues  $E_j^\pm$ differ for all $h>0$, and $j\in\Z_{\geq0}$.  Moreover, $D_je^{-\frac{d_j}{h}}\leq E_j^--E_j^+\leq C_je^{-\frac{c_j}{h}}$, for constants $c_j,d_j,C_j,D_j>0$.
\end{theorem}

\begin{figure}[h]
\centering
\includegraphics[trim={3.3cm 7cm 2.5cm 7cm},clip,width=0.35\textwidth]{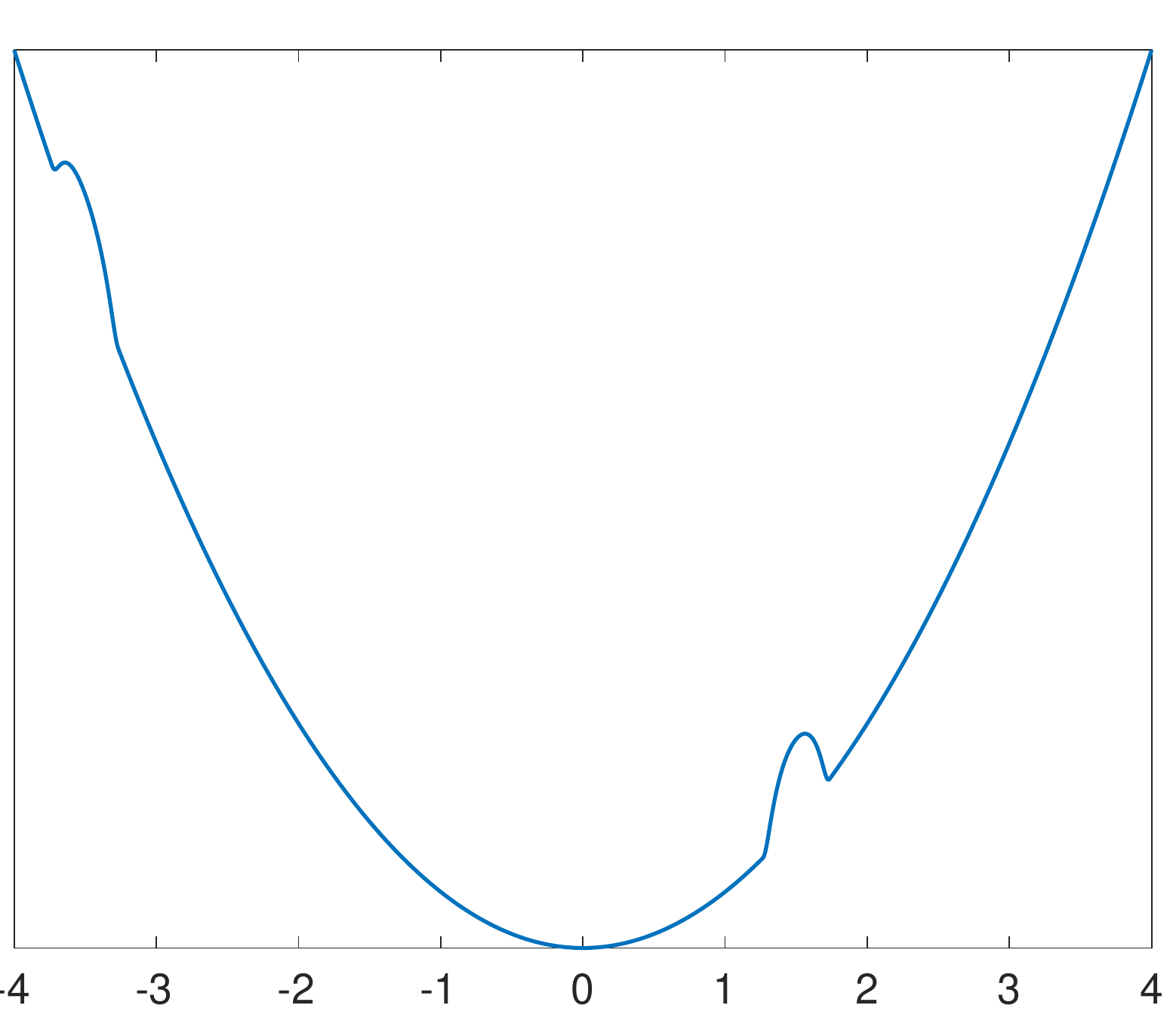}
\includegraphics[trim={3.3cm 7cm 2.5cm 7cm},clip,width=0.35\textwidth]{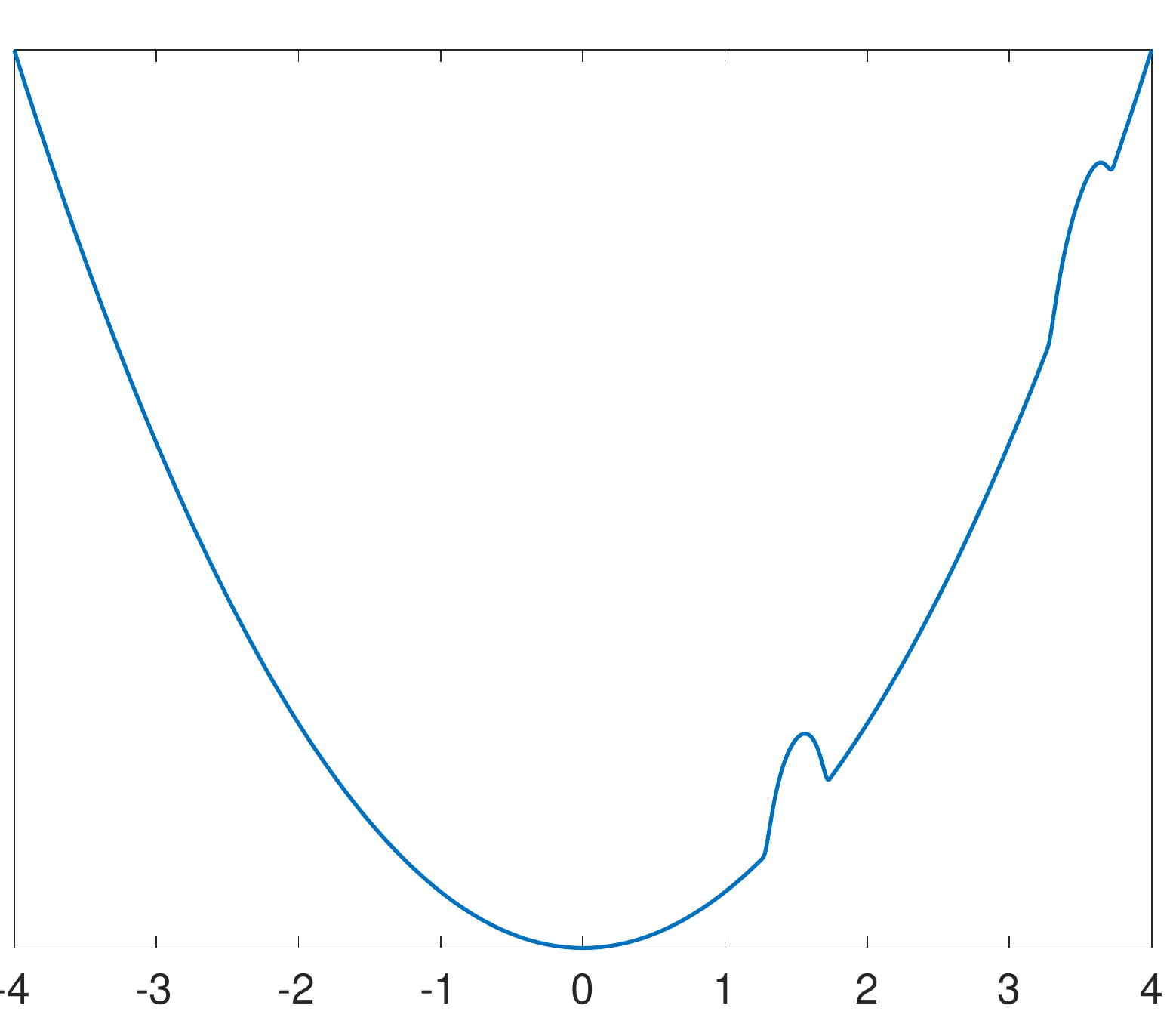}
\caption{Two potential functions $V^\pm$ whose corresponding semiclassical spectra agree up to $O(h^\infty)$, yet whose eigenvalues are disctinct ($V^-$ left, $V^+$ right).}
\label{fig:Vpm}
\end{figure}

The above theorem was conjectured by Guillmen and Hezari in \cite{guillemin2012hezari}, where they establish a similar  result for the ground state eigenvalues, except they cannot rule out the possibility of a subsequence $h_k\to0$ where the ground states agree.  Moreover, they do not produce a lower bound for the difference of ground state eigenvalues. Their methods rely on the Kato-Rellich theorem, which guarantees the analyticity of the ground state eigenvalues in $h$, to extend the nonequality of the ground state potentials at $h=1$, for all $h>0$ except for possibly a sequence $h_k\to 0$.  Our methods are distinct from the previously mentioned paper, we produce an explicit exponential lower bound for the difference of ground state eigenvalues by adapting the Agmon bounding methods employed by Simon in \cite{simon1984semiclassical}.  

To motivate our results, we consider the inverse spectral problem for the semiclassical Schr\"odinger operator: if the spectra of two semiclassical Schr\"odinger operators agree, are their potential functions isometric?  Some positive results exist, yet the methods employed can only distinguish spectra up order $O(h^N)$, for $N$ finite or $\infty$.  To account for the coarseness of these method one makes restrictions on the classes of potentials.  Datchev-Hezari-Ventura \cite{datchev2010spectral} use  order $o(h^2)$ semiclassical trace invariants to show that radial monotonic potentials are spectrally determined amongst all other potentials.  When working with analytic potentials with a unique global minimum at the origin, order $O(h^\infty)$ knowledge of the low lying eigenvalues determine the potential, see \cite{guillemin2005uribe} and \cite{hezari2009inverse}.    Colin de Verdiere in \cite{verdiere2011semi} studies the single well potential with a symmetry condition using order $o(h^2)$ techniques, and conjectures that the symmetry condition is necessary.  The main result of this paper serves as a counterexample which establishes the necessity of the symmetry condition of \cite{verdiere2011semi}.   

The corresponding limitation for the Laplacian over planar domains was initially conjectured by Zelditch in \cite{zelditch2004inverse}, and established by Fulling-Kuchment in \cite{fulling2005coincidence}.  Here, two non-isometric planar domains with the same wave trace invariants were show to exist.  The so called Penrose-Lifshits mushroom domains in Figure \ref{fig:PenroseLifshits} have been shown to posses the property that $\text{Tr}(\cos(t\sqrt{\Delta_\Omega}))-\text{Tr}(\cos(t\sqrt{\Delta_{\Omega'}}))\in C^\infty(\R)$, that is the wave traces possess the same singular structure, rendering wave trace methods unable to distinguish these domains.  
\begin{figure}[h]
\centering
\includegraphics[width=0.5\textwidth]{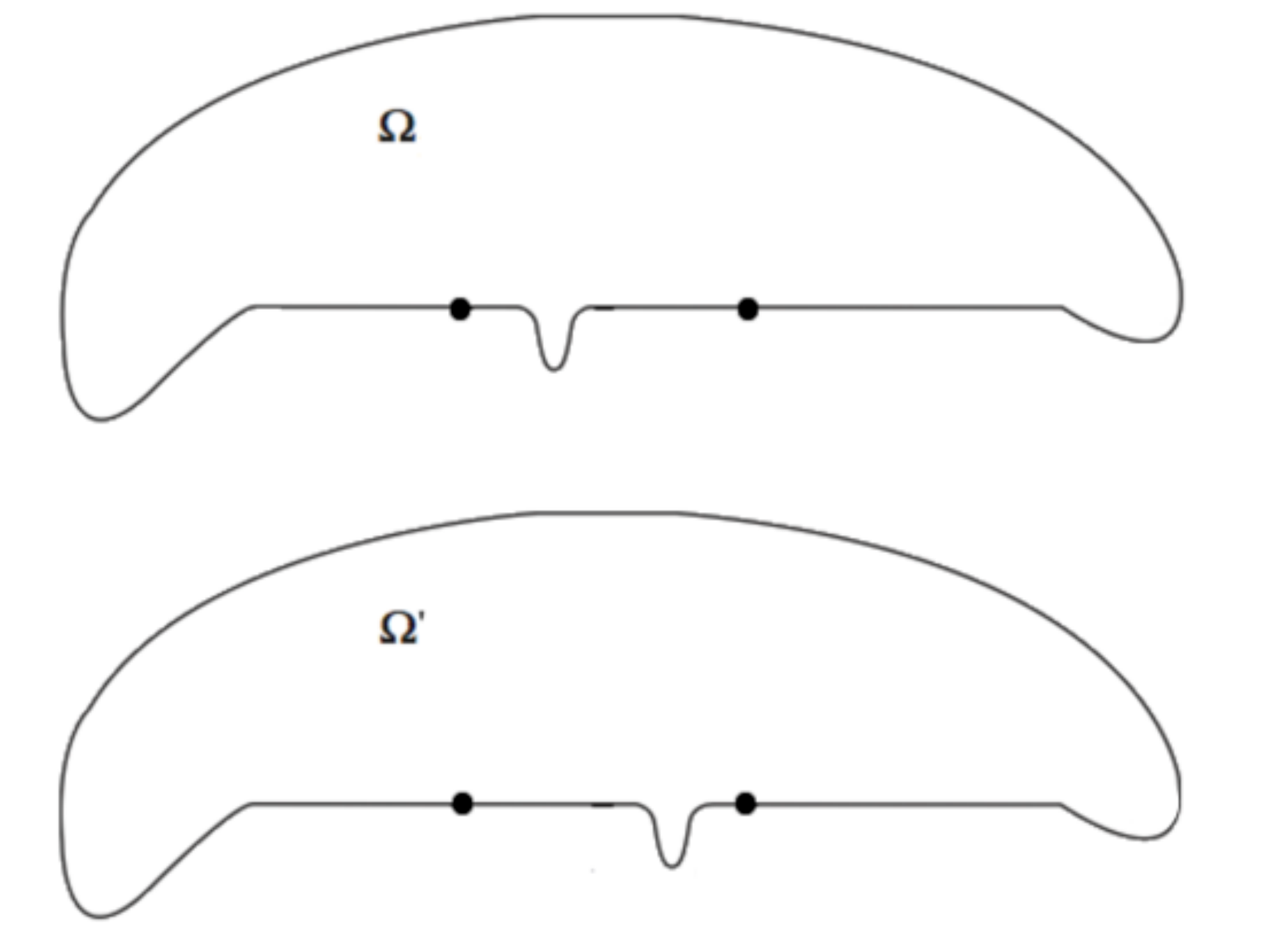}
\caption{Two non-isometric domains with the same wave trace invariants.}
\label{fig:PenroseLifshits}
\end{figure}

I would like to thank my advisor Hamid Hezari for the many insightful conversations and the suggestion to study this problem.

\section{Exponential Estimates for Eigenfunctions}
In this section we will restrict our attention to perturbations of the semiclassical harmonic oscillator:
\[
	H=-h^2\frac{d^2}{dx^2}+x^2+\gamma(x),\quad\gamma\in C_0^\infty(\R\setminus\{0\})
\]
and denote the semiclassical harmonic oscillator $H_0=-h^2\frac{d^2}{dx^2}+x^2$.  We will develop upper and lower locally uniform estimates for the ground states of these operators, which we will ultimately transfer to eigenvalue bounds using Hadamard's variational formula.  Theorem \ref{theorem:mainResult} will follow from a judicious choice of $\gamma$, analogous to the domains shown in Figure \ref{fig:PenroseLifshits}.

To begin we require the following characterization of the perturbed harmonic oscillator spectra,
\[
	\text{spec}(H)=\text{spec}(H_0)+O(h^\infty)=\{h,3h,5h,7h,\dots\}+O(h^\infty).
\]
Using the methods of quantom-birkhoff normal forms at the bottom of a potential well established by Sj\"ostrand \cite{sjostrand1992semi}, or Borh-Sommerfeld quantization to all orders given by Colin de Verdiere \cite{de2005bohr}, one can see that the above holds.  Moreover, any two perturbations of the harmonic oscillator have spectra which agree up to $O(h^\infty)$. This fact is underpinned by the observation that the level sets of the Hamiltonians corresponding to $H^\pm$ enclose the same area, see Figure \ref{fig:Hpm}.

\begin{figure}[h]
\centering
\includegraphics[trim={4.5cm 8.5cm 4cm 7.5cm},clip,width=0.35\textwidth]{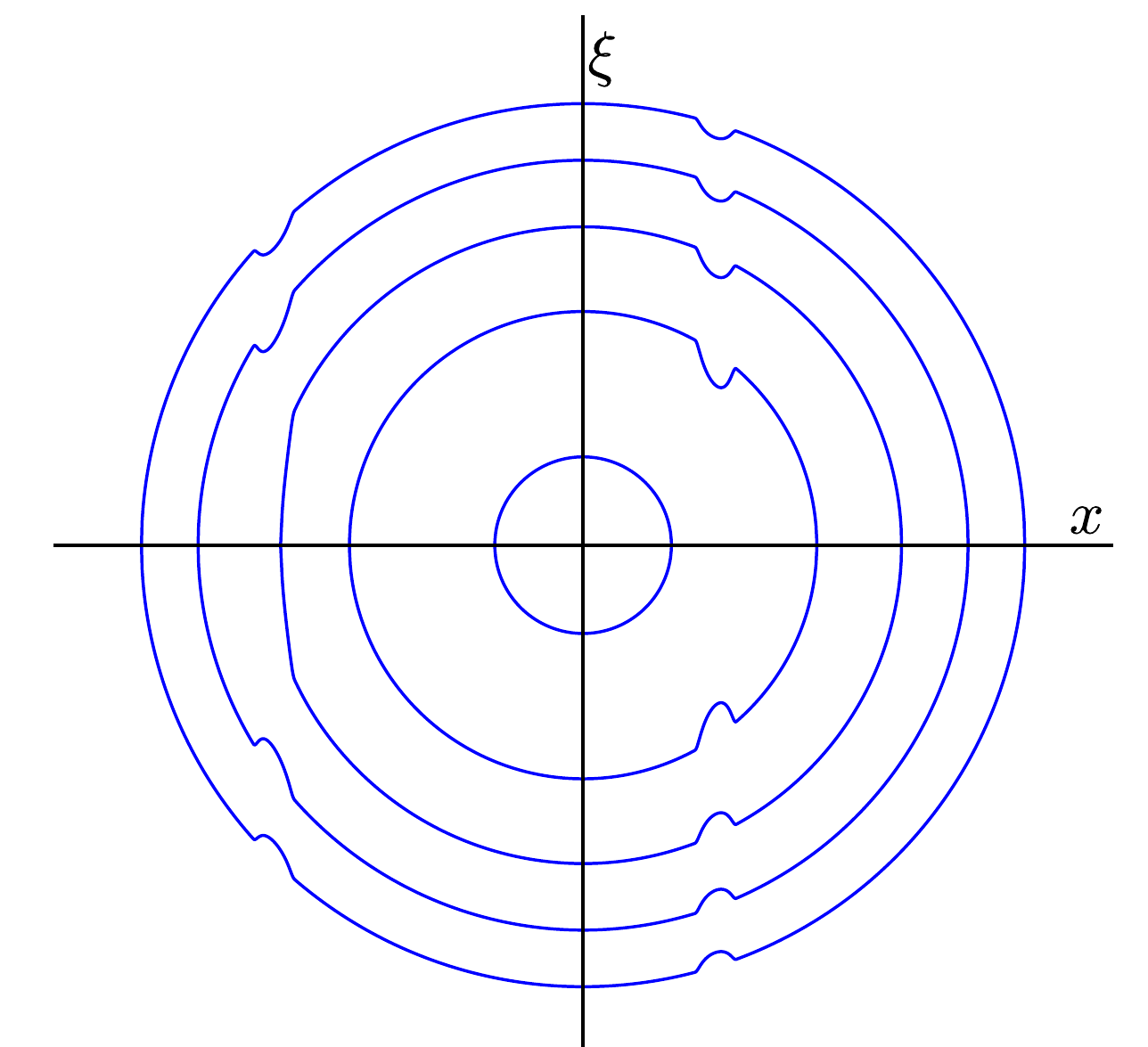}
\includegraphics[trim={4.5cm 8.5cm 4cm 7.5cm},clip,width=0.35\textwidth]{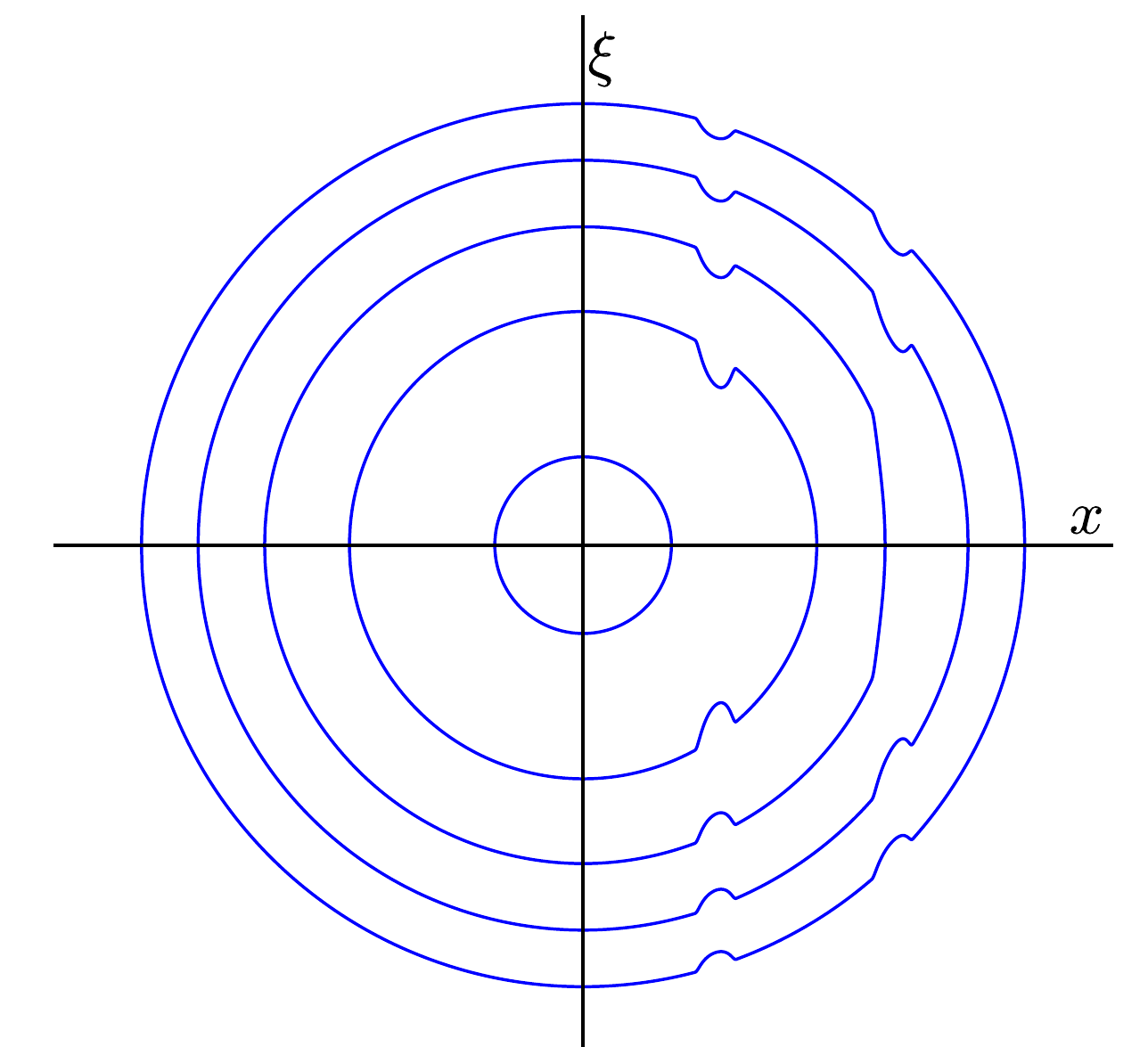}
\caption{Level sets of the Hamiltonians corresponding to $H^\pm$ ($H^-$ left, $H^+$ right).}
\label{fig:Hpm}
\end{figure}

\begin{proposition}
\label{prop:eigBound}
Let $H=-h^2\frac{d^2}{dx^2} +V(x)$ be a semiclassical Schr\"odinger operator with real valued potential of the form $V(x)=x^2+\gamma(x)$, for some $\gamma\in C_0^\infty(\R)$ supported away from the origin.
Let $\psi_j(x)$ denote the $L^2$ normalized eigenfunctions of $H$.  Then, for any $\delta,r,R\in(0,\infty)$ with $r<R$, there exist constants $C_j,D_j>0$ such that for all $h>0$ sufficiently small, the following estimates hold uniformly on $r<|x|<R$
\begin{flalign*}
	\qquad(i)&\qquad|\psi_j(x)|\leq C_je^{-\frac{1}{h}(1-\delta)^2\left|\int_0^x\sqrt{V(x)}\,dx\right|}&\\
	\qquad(ii)&\qquad|\psi_j(x)|\geq D_je^{-\frac{1}{h}(1+\delta)^2\left|\int_0^x\sqrt{V(x)}\,dx\right|}&
\end{flalign*}
\end{proposition}
\begin{remark}
\normalfont
A similar proposition appears in the work of Simon \cite{simon1984semiclassical}, with some important differences.  Simon studies the ground state eigenfunctions of the double well potential, in contrast to this paper where we study the single well potential.  Additionally, the above result holds for each eigenfunction, whereas in the paper of Simon, the analogous result was only established for the ground state.
\end{remark}
\begin{proof}[Proof of Proposition \ref{prop:eigBound}.(i)]
Define $\tilde\varphi(x)=\left|\int_0^x\sqrt{V(t)}\,dt\right|$.  Note that for each choice of positive constants $\sigma,r,R,\delta$ there exists a smooth mollifier $f$ which satisfies the following:
\begin{flalign*}
	\qquad(1)&\text{ on the annulus }r<|x|<R,\;f*\tilde\varphi\leq \tilde\varphi+\sigma\text{ and }f*(\tilde\varphi')\leq\tilde\varphi'+\sigma&\\
	\qquad(2)&\text{ there exists positive constants } a,k\text{ with }a<r,\text{ such that for all }|y|\leq a\text{ and }r<|x|<R,\text{ we have }&\\		&\text{ }\delta(f*\tilde\varphi)(x)-(f*\tilde\varphi)(y)\geq k>0.&
\end{flalign*}
Let $\epsilon>0$ and define $A=\inf\left\{(1-\epsilon)\frac{\sqrt{V(x)}}{\sqrt{V(x)}+\sigma}:r<|x|<R\right\}$.
Finally, define $\varphi=(Af)*\tilde\varphi$.  We may choose $\epsilon,\sigma$ sufficiently small so that 
\[
	(1-\delta)\tilde\varphi\leq(1-\delta)(f*\tilde\varphi)\leq A(f*\tilde\varphi),
\]
and 
\[
	A(f*\tilde\varphi)\leq A(\tilde\varphi+\sigma)\leq\tilde\varphi+\sigma\leq\tilde\varphi+\delta\tilde\varphi=(1+\delta)\tilde\varphi.
\]
Thus,
\begin{equation}
\label{eqn:mollBound}
	(1-\delta)\tilde\varphi\leq\varphi\leq(1+\delta)\tilde\varphi.
\end{equation}
Also,
\begin{equation}
\label{eqn:mollDerivBound}
	|\varphi'(x)|
	\leq A(\sqrt{V(x)}+\sigma)
	\leq(1-\epsilon)\sqrt{V(x)}.
\end{equation}
Without loss of generality we will smoothly re-define $\varphi$ such that for $|x|\geq 2 R$, $\varphi'(x)=0$.
Recall that $E_j$ is the $j^\text{th}$ eigenvalue of $H$, and let $\psi$ be any smooth function. We will now introduce the key object which gives us the upper bound of Proposition \ref{prop:eigBound}.(i):
\begin{align*}
	\left\langle e^{\frac{\varphi}{h}}\psi,(H-E_j)e^{-\frac{\varphi}{h}}\psi\right\rangle
	&=\left\langle\psi,(V-E_j)\psi\right\rangle-\left\langle\psi,\left(h\frac{d}{dx}-\varphi'\right)^2\psi\right\rangle.
\end{align*}
By the fundamental theorem of calculus, and that $\varphi'$ is eventually 0, we have that 
\[
	\left\langle\psi,\left[\left(h\frac{d}{dx}\right)\varphi'+\varphi'\left(h\frac{d}{dx}\right)\right]\psi\right\rangle=0.
\]
So
\begin{align*}
	\left\langle e^{\frac{\varphi}{h}}\psi,(H-E_j)e^{-\frac{\varphi}{h}}\psi\right\rangle
	&=\left\langle\psi,\left[\left(-h^2\frac{d^2}{dx^2}-(\varphi')^2\right)+V-E_j\right]\psi\right\rangle.
\end{align*}
By bound (\ref{eqn:mollDerivBound}), and since $-\frac{d^2}{dx^2}$ is positive definite we arrive at the following bound: 
\begin{align*}
	\left\langle e^{\frac{\varphi}{h}}\psi,(H-E_j)e^{-\frac{\varphi}{h}}\psi\right\rangle \geq\left\langle\psi,\left[\epsilon V-E_j\right]\psi\right\rangle.
\end{align*}
Recall that the eigenvalue $E_j=h(2j+1)+O(h^\infty)$, so for all $|x|>a/2$ and $h$ sufficiently small, there exist a $C>0$ satisfying $\epsilon V(x)-E_j\geq \frac{1}{C}>0$.  Taking $\psi$ supported in $|x|>a/2$, we get 
\[
	||\psi||_2^2\leq C\left\langle e^{\frac{\varphi}{h}}\psi,(H-E_j)e^{-\frac{\varphi}{h}}\psi\right\rangle.
\]

Let $\psi=e^\frac{\varphi}{h}\eta\psi_j$, where $\eta(x)=0$ if $|x|<\frac{a}{2}$ and $\eta(x)=1$ if $|x|>a$, and is smooth. Using the above bound, 
\begin{align*}
	\int_{|x|>a}e^\frac{2\varphi}{h}\psi_j^2\,dx
	&\leq C\left\langle e^\frac{2\varphi}{h}\psi_j,(H-E_j)\eta\psi_j\right\rangle.
\end{align*}
Using the eigenvalue equation, bounding the derivatives of $\eta$ on $a/2<|x|<a$, and using monotonicity of $\varphi$ on $|x|>a/2$, we get 
\begin{align*}
	\int_{|x|>a}e^\frac{2\varphi}{h}\psi_j^2\,dx
	&\leq h^2Ce^\frac{2\varphi(a)}{h}\int_{\frac{a}{2}<|x|<a} (K\psi_j'\psi_j+D\psi_j^2).
\end{align*}
Applying Cauchy-Schwarz, using that $E_j=h(2j+1)+O(h^\infty)$ and the eigenvalue equation to see that for some constant $K_j$, $||\psi_j'||=\frac{K_j}{h}+O(h^\infty)$, along with the fact that the eigenfunctions are $L^2$ normalized we have
\begin{align*}
	\int_{|x|>a}e^\frac{2\varphi}{h}\psi_j^2\,dx
	&\leq Ce^\frac{2\varphi(a)}{h}(K_jh+h^2D).
\end{align*}
Finally, for $h$ sufficiently small, we get the following estimate:
\begin{equation}
	\label{eqn:intermdUppBd}
	\int_{|x|>a}e^\frac{2\varphi}{h}\psi_j^2\,dx
	\leq e^\frac{2\varphi(a)}{h}C.
\end{equation}

We will now use bound (\ref{eqn:intermdUppBd}) to achieve the desired pointwise bound on $r<|x|<R$ by using convexity of $|\psi_j|$ in the forbidden region $|x|>\sqrt{h(2j+1)+O(h^\infty)}$.  Without loss of generality assume that $r<x<R$, then for all $b$ satisfying $0<b<r-a$ and $h$ sufficiently small we have
\begin{align*}
	\psi_j^2(x)
	&\leq\frac{1}{2b}\int_{x-b}^{x+b}e^{-2\frac{\varphi}{h}}e^{2\frac{\varphi}{h}}\psi_j^2\,dt.
\end{align*}
By the monotonicity of $\varphi$ in $|x|>a$ and applying bound (\ref{eqn:intermdUppBd}) we get
\begin{align*}
	\psi_j^2(x)
	&\leq\frac{1}{2b}e^{-2\frac{\varphi(x-b)}{h}}e^\frac{2\varphi(a)}{h}C.
\end{align*}
Let $M=\max_{a<|x|<R}\varphi'(x)$.  Then using $-\varphi(x-b)\leq -\varphi(x)+bM$ we get 
\begin{align*}
	\psi_j^2(x)
	&\leq\frac{C}{2b}e^{-2\frac{\varphi(x)}{h}}e^{2\frac{bM}{h}}e^\frac{2\varphi(a)}{h}.
\end{align*}
By our choice of $a$ which guarantees $\delta\varphi(x)-\varphi(a)\geq k>0$, we get
\begin{align*}
	\psi_j^2(x)
	&\leq\frac{C}{2b}e^{-2(1-\delta)\frac{\varphi(x)}{h}}e^{2\frac{bM-k}{h}}.
\end{align*}
Taking $b<k/M$ and using bound (\ref{eqn:mollBound}) establishes the upper bound of Proposition \ref{prop:eigBound}:
\[
	|\psi_j(x)|\leq Ce^{-(1-\delta)^2\frac{\tilde\varphi(x)}{h}}.
\]
\end{proof}


The proof of Proposition \ref{prop:eigBound}.(ii) is more difficult.  We will first need to establish a lower bound for $\psi_j$ on the boundary of the allowed region $\{-\sqrt{E_j},\sqrt{E_j}\}$.  To achieve this we need a fact about the $L^2(\R)$ convergence of rescaled eigenfunctions of $H$.

\begin{lemma}
\label{lemma:L2conv}
Let $\tilde H=-\frac{d^2}{dx^2}+x^2+h^{-1}\gamma(\sqrt{h}x)$ be the rescaled perturbed semiclassical harmonic oscillator, and $\tilde\psi_j(x)$ be the eigenfunctions of $\tilde H$.  Let $\tilde H_0=-\frac{d^2}{dx^2}+x^2$ be the harmonic oscillator, and let $\tilde\kappa_j$ denote its eigenfunctions.  Then $\tilde\psi_j(x)\to\tilde\kappa_j(x)$ in $L^2(\R)$ as $h\to 0$.
\end{lemma}
\begin{proof}
Define $(U\phi)(x):=h^{-\frac{1}{2}}\phi(h^{-\frac{1}{2}}x)$, and notice that $\tilde H=h^{-1}U^{-1}HU$ (and $\tilde H_0=h^{-1}U^{-1}H_0U$).  Let $\tilde E_j$ denote the eigenvalues of $\tilde H$.  The rescaling operator $U$ gives us a relationship between the eigenfunctions and eigenvalues of $\tilde H$ and $H$:
\[
	E_j=h\tilde E_j,\quad \psi_j=U\tilde\psi_j.
\] 
From this, and the fact that $\text{spec}(H)=\{h,3h,5h,\dots\}+O(h^\infty)$, we get that
\[
	\text{spec}(\tilde H)=\{1,3,5,\dots\}+O(h^\infty).
\]
Note by a similar rescaling, we get that $\text{spec}(\tilde H_0)=\{1,3,5,7,\dots\}$.

We define the projections $P_{\tilde H}$ and $P_{H_0}$ as follows:
\[
	P_{\tilde H}=\frac{1}{2\pi i}\oint_{\partial D(\tilde E_j,\epsilon)}(z-\tilde H)^{-1}\,dz,\qquad P_{H_0}=\frac{1}{2\pi i}\oint_{\partial D(2j+1,\epsilon)}(z-H_0)^{-1}\,dz,
\]
where $D(z,\epsilon)$ is the disc centered at $z$ of radius $\epsilon$.  Note that these projections are rank one, onto the subspaces spanned by $\tilde\psi_j$ and $\tilde\kappa_j$ respectively, see \cite{berezin2012schrodinger}.  Now, for all $\epsilon>0$ and $h>0$ sufficiently small, we have that $|\tilde E_{j-1}-(2j-1)|<\epsilon$, $|\tilde E_j-(2j+1)|<\epsilon$, and $|\tilde E_{j+1}-(2j+3)|<\epsilon$ so we may write the difference of the projections under the same integrand:

\[
	P_{\tilde H_0}-P_{\tilde H}
	=\frac{1}{2\pi i}\oint_{\partial D(2j+1,\epsilon)}(z-\tilde H_0)^{-1}(\tilde H_0-\tilde H)(z-\tilde H)^{-1}\,dz.
\]

Now, applying this difference to the eigenfunction $\tilde \psi_j$ yields:
\begin{align*}
	\norm{(P_{\tilde H_0}-P_{\tilde H})\tilde\psi_j}_2
	&\leq\frac{1}{2\pi}\norm{\oint_{\partial D(2j+1,\epsilon)}(z-\tilde E_1(h))^{-1}(z-H_0)^{-1}\,dz}_{op}\norm{(h^{-1}\gamma(\sqrt h x))\tilde\psi_j}_2.
\end{align*}
Bounding by the distance to the spectrum along the contour, and that $\gamma\in C_0^\infty(\R\setminus\{0\})$, we get 
\begin{align*}
	\norm{(P_{\tilde H_0}-P_{\tilde H})\tilde\psi_j}_2
	&\leq C\norm{h^{-1}\tilde \psi_j}_{L^2(\supp(\gamma(\sqrt h x)))}.
\end{align*}
One can see that  $\lim_{h\to 0}\norm{h^{-1}\tilde \psi_j}_{L^2(\supp(\gamma(\sqrt h x)))}=0$  by Proposition  \ref{prop:eigBound}.(i), and that $\lim_{h\to 0}\inf \supp(\gamma(\sqrt h x))=\infty$.  Some algebra yields the desired results: 
\[
	1-\langle\tilde\psi_j,\tilde\kappa_j\rangle^2
	=\norm{(P_{\tilde H_0}-P_{\tilde H})\tilde\psi_j}_2^2,
\]
so $\lim_{h\to 0}\langle\tilde\psi_j,\tilde\kappa_j\rangle=1$.  Finally
\[
	\norm{\tilde\psi_j-\tilde\kappa_j}_2^2
	=2(1-\langle\tilde\psi_j,\tilde\kappa_j\rangle),
\]
which tends to 0 in the limit as $h\to 0$.
\end{proof}

The above lemma tells us that in the semiclassical limit, the $L^2$ mass of the perturbed eigenfunctions are distributed in the same way as the non-perturbed eigenfunctions.  In fact, a stronger result holds, the perturbed eigenfunctions converge locally uniformly to the non-perturbed eigenfunctions.  This is proven along the way to establishing lemma \ref{lemma:boundClassical}.  However, before we can prove the lemma, we require a bound on the roots of $\tilde\kappa_j$.

\begin{lemma}
\label{lemma:boundHermite}
For each $j\in\Z_{\geq 0}$, $\max\{|r|:\tilde\kappa_j(r)=0\}\leq\sqrt{2j-2}$.
\end{lemma}
\begin{proof}
Let $A^*:=(x-\frac{d}{dx})$ and $A:=(x+\frac{d}{dx})$ be the creation and annihilation operators associated to $\tilde H_0$, satisfying $AA^*-\text{Id}=A^*A+\text{Id}=\tilde H_0$.  One can see that $\tilde\kappa_j=\frac{1}{\sqrt{2^j j!}}{A^*}^j\tilde\kappa_0$, where $\tilde\kappa_0=\pi^{-\frac{1}{4}}e^{-\frac{x^2}{2}}$.  Moreover, 
\[
	\tilde\kappa_j(x)=\frac{1}{\sqrt{2^j j!\sqrt{\pi}}}P_j(x)e^{-\frac{x^2}{2}},
\]
where $P_j(x)$ are the physicist's Hermite polynomials.  From the creation operator definition of $\tilde\kappa_j$, one can see that the $P_j$ satisfy the recurrence relation $P_{j+1}(x)=2xP_j(x)-\frac{d}{dx}P_j(x)$.  An inductive argument shows that $\frac{d}{dx}P_j(x)=2jP_{j-1}(x)$, and so we have the modified recurrence relation $P_{j+1}(x)=2xP_j(x)-2jP_{j-1}(x)$.  Define the tridiagonal symmetric matricies
\[
	C_j=\begin{pmatrix}
	0 & \sqrt{\frac{1}{2}} & 0 & 0 & 0 &\cdots & 0& 0\\
	\sqrt{\frac{1}{2}} & 0 & \sqrt{\frac{2}{2}} & 0 & 0 & \cdots & 0& 0\\
	0 & \sqrt{\frac{2}{2}} & 0 & \sqrt{\frac{3}{2}} & 0 & \cdots & 0& 0\\
	\vdots &  &  & \ddots & & & \vdots & \vdots\\
	0 & 0 & 0 & 0 & 0 & \cdots & 0 & \sqrt{\frac{j-1}{2}}\\
	0 & 0 & 0 & 0 & 0 & \cdots & \sqrt{\frac{j-1}{2}}& 0
	\end{pmatrix}.
\]
Through induction and using the above recurrence relation, one can see that $P_j(x)=2^j\text{det}(xI-C_j)$, and so the roots of $P_j(x)$ are exactly the eigenvalues of $C_j$.  Applying the Gershgorin circle theorem, we may bound the eigenvalues of $C_j$, hence the roots of $P_j$ which are exactly the roots of $\tilde\kappa_j$:
\[
	\max\{|r|:\tilde\kappa_j(r)=0\}\leq \sqrt{\frac{j-2}{2}}+\sqrt{\frac{j-1}{2}}\leq\sqrt{2j-2}.
\]
\end{proof}


\begin{lemma}
\label{lemma:boundClassical}
For all $h$ sufficiently small, $|\psi_j(\pm\sqrt{E_j})|\geq \frac{D_j}{\sqrt h}\geq 1$.
\end{lemma}
\begin{proof}
We will first get the bound on the rescaled eigenfunction $\tilde\psi_j$, then transfer this bound to $\psi_j$.  Let $\theta=\tilde\psi_j-\tilde\kappa_j$.  Lemma \ref{lemma:L2conv} establishes that $\theta\to 0$ in $L^2(\R)$.  Now, let $I\subseteq \R$ be any bounded interval.  Then $\theta''\to 0$ in $L^2(I)$, since 
\begin{align*}
	\norm{\theta''}_{L^2(I)}
	&\leq\norm{h^{-1}\gamma(\sqrt{h}x)\tilde\psi_j(x)}_{L^2(\R)}+\norm{x^2(\tilde\psi_j(x)-\tilde\kappa_j(x))}_{L^2(I)}+\norm{\tilde E_j\tilde\psi_j(x)-\tilde\kappa_j(x)}_{L^2(\R)},
\end{align*}
where the first term vanishes by Proposition \ref{prop:eigBound}.(i) and that $\lim_{h\to 0}\inf \supp(\gamma(\sqrt h x))=\infty$, the second term vanishes due to Lemma \ref{lemma:L2conv}, and the third term vanishes due to lemma \ref{lemma:L2conv} and since $\lim_{h\to 0}\tilde E_j=2j-1$.

Let $\eta\in C_0^\infty(\R)$ and consider the following identity: $\eta(\theta')^2+\eta\theta\theta''-\frac{1}{2}\theta^2\eta''=\frac{d}{dx}(\eta\theta\theta'-\frac{1}{2}\theta^2\eta')$. Integrating yields
\[
	\int\eta(\theta')^2\,dx=\frac{1}{2}\int\eta''\theta^2\,dx-\int\eta\theta\theta''\,dx,
\]
which tends to 0 as $h\to 0$ since $\theta,\theta''\to 0$ in $L^2(I)$, thus $\theta'\to 0$ in $L^2(I)$.  Now let $x\in [b,c]=I$ .  Then, for each $a<b$, by the fundamental theorem of calculus
\[
	|\theta(x)-\theta(a)|\leq||\theta'||_{L^2([a,c])},
\]
which tends to zero as $h\to 0$.  Since $\lim_{a\to-\infty}\theta(a)=0$, we must have that $\lim_{h\to 0}||\theta(x)||_{L^\infty(I)}=0$, or equivalently $\tilde\psi_j\to\tilde\kappa_j$ in $L^\infty(I)$.

Now, we pass to the unscaled semiclassical operator.  Working in the allowed region $C=\left[-\sqrt{E_j},\sqrt{E_j}\right]$: 
\begin{align*}
	\norm{\psi_j(x)-h^{-\frac{1}{2}}\tilde\kappa_j(h^{-\frac{1}{2}}x)}_{L^\infty(C)}
	&=h^{-\frac{1}{2}}\norm{\tilde\psi_j(x)-\tilde\kappa_j(x)}_{L^\infty\left(\left[-\sqrt{\tilde E_j},\sqrt{\tilde E_j}\right]\right)}.
\end{align*}
Multiplying the above equation by $\sqrt{h}$, and recalling that $\sqrt{\tilde E_j}=\sqrt{2j+1}+O(h^\infty)$ yields
\begin{align*}
	\lim_{h\to 0}\norm{\sqrt{h}\psi_j(x)-\tilde\kappa_j(h^{-\frac{1}{2}}x)}_{L^\infty(C)}
	&=0.
\end{align*}
Lemma \ref{lemma:boundHermite} establishes that $|\tilde\kappa_j(\pm\sqrt{2j+1})|>0$, and since $\sqrt{E_j}=\sqrt{h(2j+1)}+O(h^\infty)$ for all $h$ sufficiently small there is a constant $D_j>0$ satisfying, 
\[
	|\tilde\kappa_j(\pm h^{-\frac{1}{2}}\sqrt{E_j})|\geq D_j.
\]
Thus by the above limit, for sufficiently small $h$, $|\psi_j(\pm\sqrt{E_j})|\geq\frac{D_j}{\sqrt h}$.

\end{proof}
The final lemma that we require gives us a way to transfer the lower bound of Lemma \ref{lemma:boundClassical} outward from the boundary of the  allowed region $[-\sqrt{E_j},\sqrt{E_j}]$.

\begin{lemma}
\label{lemma:boundEndpoint}
Consider the interval $I=[x_1,x_2(1+\epsilon)]$ with $\sqrt{E_j}\leq x_1$.  Let $v=\sup_{t\in I}\sqrt{V(t)}$.  Then 
\[
	|\psi_j(x_2)|\geq e^{-\frac{v}{h}(x_2-x_1)}\left(1-e^{-2\epsilon\frac{v}{h}x_2}\right)|\psi_j(x_1)|.
\]
\end{lemma}
\begin{proof}
Define $\phi(x)=\left(e^{-\frac{v}{h} (x-x_1)}-e^{-2\frac{v}{h} (x_2(1+\epsilon)-x_1)}e^{\frac{v}{h}(x-x_1)}\right)|\psi_j(x_1)|$.  Notice that $\phi(x_2(1+\epsilon))=0$, and 
\[
	\phi(x_1)=(1-e^{-2\frac{v}{h} (x_2(1+\epsilon)-x_1)})|\psi_j(x_1)|\leq|\psi_j(x_1)|.
\]
Thus, on $\partial I$, $\phi\leq|\psi_j|$.  We claim that on the interior of $I$, $\phi\leq|\psi_j|$.  Let $\eta =|\psi_j|-\phi$ and define $I_0=\{x:\eta(x)<0\}\subseteq I$.  Now, since $\eta\geq0$ on $\partial I$, $\eta=0$ on $\partial I_0$.  Let $W(x)=V(x)-E_j$ and note by choice of lower bound on $I$ we have that $W\geq 0$ on $I$. Then on $I_0$:
\[
	\frac{d^2}{dx^2}\eta
	=W\eta-\left(\left(\frac{v}{h}\right)^2-W\right)\phi,
\]
which is less than zero for all $h$ sufficiently small.  Thus, $\eta$ is concave on $I_0$, and so it attains is minimum on $\partial I_0$.  So $\eta\geq 0$ on $I_0$, and $I_0$ is  empty, i.e. $\phi\leq|\psi_j|$ on $I$.  With this, evaluating $\phi$ at $x_2$ yields $|\psi_j(x_2)|\geq e^{-\frac{v}{h} (x_2-x_1)}(1-e^{-2\epsilon\frac{v}{h} x_2})|\psi_j(x_{1})|$.
\end{proof}

Finally, we are able to prove Proposition \ref{prop:eigBound}.(ii), which will involve propagating the bound from Lemma \ref{lemma:boundClassical}, using Lemma \ref{lemma:boundEndpoint}.
\begin{proof}[Proof of Proposition \ref{prop:eigBound}.(ii)]
Let $\sqrt{E_j}<|x|<R$, and without loss of generality, assume $x$ is positive.  Let $0=x_0<x_1=\sqrt{E_j}<x_2<\cdots<x_n=x$ be a partition of $[0,x]$ such that the upper Darboux sum of $\int_0^x\sqrt{V(x)}\,dx$ satisfies
\[
	\sum_{i=1}^n\sup_{t\in[x_{i-1},x_i]}\sqrt{V(t)}|x_i-x_{i-1}|\leq \left(1+\frac{\delta}{2}\right)\int_0^x\sqrt{V(t)}\,dt.
\]
for some $\delta>0$.  Let $\epsilon>0$ be small enough such that if we define $D_i=[x_{i-1},x_i(1+\epsilon))]$, and $v_i=\sup_{t\in D_i}\sqrt{V(t)}$, we get
\begin{equation}
\label{eqn:extdSum}
	\sum_{i=1}^nv_i|x_i-x_{i-1}|\leq \left(1+\delta\right)\int_0^x\sqrt{V(t)}\,dt.
\end{equation}
Next, applying Lemma \ref{lemma:boundEndpoint} iteratively on this partition, we get that 
\begin{align*}
	|\psi_j(x)|
	&\geq e^{-\sum_{i=2}^n\frac{v_i}{h}(x_i-x_{i-1})}\prod_{i=2}^n\left(1-e^{-2\epsilon\frac{v_i}{h}x_i}\right)|\psi_j(x_1)|.
\end{align*}
Lemma \ref{lemma:boundClassical} guarantees (for $h>0$ sufficiently small), $|\psi_j(x_1)|=|\psi_j(\sqrt{E_j})|\geq 1$, so we recover the full sum
\begin{align*}
	|\psi_j(x)|
	&\geq e^{-\sum_{i=1}^n\frac{v_i}{h}(x_i-x_{i-1})}\prod_{i=1}^n\left(1-e^{-2\epsilon\frac{v_i}{h}x_i}\right).
\end{align*}
Since $v_i=O(1)$ for $i\neq1$, and $v_1=O(\sqrt h)$, for all $h$ sufficiently small, we estimate $\prod_{i=1}^n\left(1-e^{-2\epsilon\frac{v_i}{h}x_i}\right)\geq2^{-n}$.  Using this along with bound (\ref{eqn:extdSum}), we get the final desired bound:
\begin{align*}
	|\psi_j(x)|
	&\geq De^{-(1+\delta)^2\frac{1}{h}\left|\int_0^x\sqrt{V}\,dt\right|}.
\end{align*}
 To get uniformity in $r<|x|<R$, extend the partition.
\end{proof}

\section{Proof of the Main Result}
To prove Theorem \ref{theorem:mainResult} we will construct the desired potential functions $V^\pm$.  Let $\alpha,\beta\in C_0^\infty(\R)$ with $\supp(\alpha)\subseteq(1,2)$ and $\supp(\beta)\subseteq(3,4)$ and set $V^\pm(x)=x^2+\alpha(x)+\beta(\pm x)$.   To recover the sub $O(h^\infty)$ differences between the eigenvalues of the associated operators $H^\pm=-h^2\frac{d^2}{dx^2}+x^2+\alpha(x)+\beta(\pm x)$, we will apply a variation in $\beta(\pm x)$. Define the following family of operators
\begin{align*}
	H_t^\pm&=-h^2\frac{d^2}{dx^2}+x^2+\alpha(x)+t\beta(\pm x),
\end{align*}
and notice that $H_1^\pm = H^\pm$.  Denote the corresponding eigenfunctions and eigenvectors $\psi_j^\pm(t,x),\,E_j^\pm(t)$.

\begin{proof}[Proof of Theorem \ref{theorem:mainResult}]
The following equation, known as Hadamard's variational formula, will be useful
\[
	\frac{d}{dt}E_j^\pm(t)=\int\beta(\pm x)(\psi_j^\pm(t,x))^2\,dx.
\]
The proof of this formula for the one dimensional case is elementary, see \cite{guillemin2012hezari}. Now, with Hadamard's variational formula and the fundamental theorem of calculus, we have
\[
	E_j^\pm(t)-E_j^\pm(0)
	=\int_0^t\int\beta(\pm x)(\psi_j^\pm(s,x))^2\,dx\,ds.
\]
Using Proposition \ref{prop:eigBound} and the above equation, and that $E_j^\pm(1)=E_j^\pm$ gives us the following bounds:
\begin{equation}
\label{eqn:lowEigDiffBd}
	E_j^\pm-E_j^\pm(0) 
	\geq\int_0^1\int\beta(\pm x)De^{-\frac{(1+\delta)^2}{h}\left|\int_0^x\sqrt{a^2+\alpha(a)+s\beta(\pm a)}\,da\right|}\,dx\,ds
\end{equation}
\begin{equation}
\label{eqn:upEigDiffBd}
	E_j^\pm-E_j^\pm(0) 
	\leq\int_0^1\int\beta(\pm x)Ce^{-\frac{(1-\delta)^2}{h}\left|\int_0^x\sqrt{a^2+\alpha(a)+s\beta(\pm a)}\,da\right|}\,dx\,ds
\end{equation}
Notice that the minus version of lower bound (\ref{eqn:lowEigDiffBd}) avoids integrating in the exponent through the support of $\alpha$.  This allows us to express the difference of the minus version of (\ref{eqn:lowEigDiffBd}) with the plus version of (\ref{eqn:upEigDiffBd}) as
\begin{align*}
	E_j^--E_j^+
	&\geq\int_0^1\int\beta(x)\left(De^{-\frac{(1+\delta)^2}{h}\int_0^x\sqrt{a^2+s\beta(a)}\,da}-Ce^{-\frac{(1-\delta)^2}{h}\int_0^x\sqrt{a^2+\alpha(a)+s\beta(a)}\,da}\right)\,dxds
\end{align*}
Factoring, recalling that $\beta$ is supported in $(3,4)$, and recognizing that the integrals in the exponent are monotonic in $x$ gives
\begin{align*}
	E_j^--E_j^+
	&\geq\int_0^1 e^{-\frac{(1+\delta)^2}{h}\int_0^4\sqrt{a^2+s\beta(a)}\,da}\int\beta(x)\left(D-Ce^{-\frac{1}{h}\int_0^x(1-\delta)^2\sqrt{a^2+\alpha(a)+s\beta(a)}-(1+\delta)^2\sqrt{a^2+s\beta(a)}\,da}\right)\,dxds.
\end{align*}
The bounds of Proposition \ref{prop:eigBound} hold for any choice of $\delta$, so we may take $\delta$ small enough to satisfy for each $x\in (3,4)$
\[
	\int_0^x(1-\delta)^2\sqrt{a^2+\alpha(a)+s\beta(a)}-(1+\delta)^2\sqrt{a^2+s\beta(a)}\,da>0.
\]
With this choice of delta, and again using monotonicity of the integral in the exponent, we have
\begin{align*}
	E_j^--E_j^+
	&\geq e^{-\frac{(1+\delta)^2}{h}\int_0^4\sqrt{a^2+\beta(a)}\,da}\left(D-Ce^{-\frac{1}{h}\int_0^3(1-\delta)^2\sqrt{a^2+\alpha(a)}-(1+\delta)^2\sqrt{a^2}\,da}\right)\int\beta(x)\,dx.
\end{align*}
For all $h$ sufficiently small, this then reduces to the inequality $E_j^--E_j^+\geq De^\frac{-d}{h}$, for some positive constants $c,C>0$.  To prove the upper bound $E_j^--E_j^+\leq Ce^\frac{-c}{h}$, a similar argument is used where we instead take the difference of the minus version of (\ref{eqn:upEigDiffBd}) with the positive version of (\ref{eqn:lowEigDiffBd}).
\end{proof}

\bibliographystyle{alpha}
\bibliography{refs}
\end{document}